\documentclass[conference]{IEEEtran}
\IEEEoverridecommandlockouts
\usepackage{cite}
\usepackage{amsmath,amssymb,amsfonts}
\usepackage{amsthm}
\usepackage{algorithmic}
\usepackage{graphicx}
\usepackage{textcomp}
\usepackage{xcolor}
\usepackage{tikz}
\usepackage{comment}
\usetikzlibrary{arrows.meta, positioning}
\def\BibTeX{{\rm B\kern-.05em{\sc i\kern-.025em b}\kern-.08em
    T\kern-.1667em\lower.7ex\hbox{E}\kern-.125emX}}

\newtheorem{theorem}{Theorem}
\newtheorem{remark}{Remark}

\begin{document}

\title{{Stochastic Delayed Dynamics of Rumor Propagation with Awareness and Fact-Checking}\\
}
\author{%
\thanks{Published in Proc. IEEE WiDS PSU 2026. DOI 10.1109/WiDSPSU69342.2026.00042.}%
\IEEEauthorblockN{1\textsuperscript{st} Lamia Alyami}
\IEEEauthorblockA{\textit{Department of Mathematics,}\\ 
\textit{College of Sciences and Arts,}\\ Najran University, P.O. Box 1988, \\ 
Najran 11001, Saudi Arabia.\\
Email: lmlasloom@nu.edu.sa\\
}

\and

\IEEEauthorblockN{2\textsuperscript{nd} Anis Hamadouche}
\IEEEauthorblockA{\textit{The School of Engineering \& Physical Sciences} \\
\textit{Heriot-Watt University}\\
Edinburgh EH14 4AS, UK \\
anis.hamadouche@hw.ac.uk}

\and

\IEEEauthorblockN{3\textsuperscript{rd} Amir Hussain}
\IEEEauthorblockA{\textit{SDAIA-KFUPM Joint Research Centre for Artificial Intelligence} \\
\textit{King Fahd University of Petroleum and Minerals}\\
Dhahran, Saudi Arabia  \\
amir.hussain@kfupm.edu.sa}

}

\maketitle

\begin{abstract}
This paper presents a stochastic delayed differential model for rumor propagation during infodemic that incorporates human behavioral response, public skepticism and fact-checking mechanisms. A discrete time delay is introduced to model natural lags in information processing and institutional response. Additionally, we adopt additive stochastic perturbations to model random fluctuations in social interaction and exposure. We present a rigorous stability analysis of the proposed rumor transmission model and derive convergence guarantees under reproduction number conditions. We also validate the model by numerical simulations and analyze the outbreak severity and quantify uncertainty under variable information processing delays. The results highlight the importance of timely awareness and fact-checking interventions for mitigating misinformation spread during pandemics
\end{abstract}


\section{Introduction}

Infodemic is defined as an overwhelming surplus of information—both online and offline—that includes inaccurate or deceptive content. This phenomenon significantly complicates medical emergency responses by obscuring reliable facts. Consequently, people may seek out dangerous, unproven remedies or reject life-saving interventions like vaccines. Ultimately, infodemic can damage public confidence in science and healthcare, creating severe risks to global health~\cite{world2024infodemic, tangcharoensathien2020framework, tomes2022historical, gisondi2022deadly, do2022infodemics}.

During the COVID-19 pandemic, a "perfect storm" of infodemics emerged. This situation arose when an unprecedented volume of data intersected with significant information deficits. Consequently, unsubstantiated claims regarding the virus's origins and various conspiracies about vaccination efforts and state mandates propagated quickly through social media and digital platforms. The widespread distribution of false information regarding the virus and its associated dangers prompted many individuals to adopt behaviors that jeopardized both personal and community safety. Furthermore, there is growing concern that this "infodemic" of misinformation is spilling over, potentially undermining public health efforts for other infectious diseases like measles and the flu~\cite{caceres2022impact}. It is therefore crucial to manage infodemics by fostering critical thinking and promoting accurate information

Infodemiology, which analyzes electronic health information, is gaining momentum as a crucial area of research. Investing in Infodemic Management (IM) offers dual benefits: it provides immediate protection for public health while simultaneously building a robust, evidence-based foundation to inform future, action-oriented strategies, tools, and protocols for managing information crises~\cite{ishizumi2024beyond}. In the absence of robust, transparent information management, health authorities risk losing public confidence, as unchecked rumors fuel dangerous behaviors. Mis/disinformation can prolong and intensify the impact of crises, leading to severe human suffering and economic devastation~\cite{world2024infodemic}.

Misinformation online is defined by its rapid, viral, and far-reaching nature, often outpacing the spread of accurate information. Several key drivers contribute to this phenomenon, including human behaviors like confirmation bias, high emotional engagement, and a preference for novelty. Additionally, platform algorithms designed for engagement, alongside the use of social bots, amplify the dissemination of false content~\cite{govindankutty2024epidemic}.

The dynamics of misinformation propagation in social networks resemble epidemiological models, where 'infection' represents exposure to and sharing of false information. This analogy has led to various mathematical models attempting to capture and predict the spread of misinformation, drawing parallels with disease transmission models~\cite{daley1964epidemics,
daley1965stochastic, sun2021uncertain, li2021social, zhao2021detecting, czerniak2023scoping}. Understanding these interactions is critical for designing successful misinformation prevention techniques. These measures include increasing digital literacy, upgrading platform algorithms to detect and prevent the dissemination of misleading information, and conducting targeted interventions to disrupt the cycle of misinformation transmission.

Even though several mathematical epidemic models have evolved, most still need to explain the human nature of selection, skepticism and social influence. The current models also fail to account for the reality that people may respond with significant delay as a function of information processing. Moreover, they also fail to quantify the conditions for equilibrium under stochastic rumor spread dynamics. Considering this, the following contributions were made to this work.

\begin{itemize}
    \item We propose an extended SEIR-type rumor model incorporating skeptical
    and fact-checked compartments to capture behavioral resistance and
    verification effects.
    \item A discrete time delay is introduced to represent delayed information
    processing and institutional response.
    \item Stochastic perturbations are incorporated to model random fluctuations
    in rumor exposure and social interaction.
    \item Equilibrium and local stability analyses are developed, leading to a
    threshold condition based on the basic reproduction number.
    \item Monte Carlo simulations and an ablation study are performed to
    investigate the combined effects of delay and reproduction number on rumor
    outbreak severity.
\end{itemize}

The remainder of this paper is organized as follows. Section~II presents the proposed stochastic delay model. Section~III
provides equilibrium and stability analysis. Section~IV presents numerical
experiments and an ablation study. Finally, Section~V concludes the paper and
discusses future research directions.


\section{Proposed Process Model}
Motivated by misinformation spread delayed dynamics (delays in institutional responses, heterogeneous individual reactions) and uncertainty in social interactions we propose a new model that incorporates information delay, behavioral response and public awareness mechanisms. The proposed stochastic delay differential equation model (SDDE) extends classical epidemic-inspired frameworks under the following assumptions:

\begin{itemize}
    \item The study assumes a demographically and geographically closed population, where no individuals enter or leave the system.
    
    \item Assuming homogeneous social structure, the model adopts a mass-action principle, where the transmission rate is directly proportional to the density of contacts.
    
    \item The transition from hearing a rumor to spreading it is not instantaneous; individuals first enter an "exposed" phase, where they pause to verify the information and think through its implications.
    
    \item Some individuals who spread rumors eventually lose interest and stop, while others become skeptical and adopt a stance of passive resistance.
    
    \item Those who doubt a rumor might eventually confirm its inaccuracy via reliable sources, shifting into a fact-checked category of people who actively oppose misinformation.
    
    \item Because public awareness campaigns and response measures are not instantaneous, a discrete time lag ($\tau>0$) is used to model this significant delay in effect.
    
    \item Real social networks are characterized by unavoidable, unpredictable fluctuations in user behavior and interaction, which are represented via Brownian motion-based stochastic perturbations.
\end{itemize}


Accordingly, the total population is divided into six mutually exclusive
compartments: $S(t)$ denotes the susceptible individuals who have not yet
encountered the rumor, $E(t)$ represents the exposed individuals who have
encountered the rumor but are not yet spreading it, $I(t)$ corresponds to the
active spreaders who disseminate the rumor, $R(t)$ describes the removed
individuals who no longer propagate the rumor, $I_g(t)$ denotes the skeptical
individuals who refrain from spreading the rumor, and $F(t)$ represents the
fact-checked individuals who have verified the information and actively reject
the rumor.

The total population satisfies the conservation law
\begin{equation}
S(t) + E(t) + I(t) + R(t) + I_g(t) + F(t) = N,
\qquad t\geq 0.
\end{equation}

\subsection{Stochastic Delay Differential Equation Model}

To capture both response delay and uncertainty in rumor diffusion, we formulate the following stochastic delay differential system:
\begin{equation}
\begin{aligned}
dS(t) &= \Big[-\beta S(t) I(t-\tau)\Big]dt 
        + \sigma_S S(t)\, dW_S(t), \\
dE(t) &= \Big[\beta S(t) I(t-\tau) - \sigma E(t)\Big]dt 
        + \sigma_E E(t)\, dW_E(t), \\
dI(t) &= \Big[\sigma E(t) - (\gamma+\rho)I(t)\Big]dt 
        + \sigma_I I(t)\, dW_I(t), \\
dR(t) &= \Big[\gamma I(t)\Big]dt 
        + \sigma_R R(t)\, dW_R(t), \\
dI_g(t) &= \Big[\rho I(t) - \theta I_g(t)\Big]dt 
        + \sigma_{Ig} I_g(t)\, dW_{Ig}(t), \\
dF(t) &= \Big[\theta I_g(t)\Big]dt 
        + \sigma_F F(t)\, dW_F(t),
\end{aligned}
\label{eq:proposed_model}
\end{equation}
where $\tau>0$ denotes the discrete information delay, and 
$\{W_S(t),W_E(t),W_I(t),W_R(t),W_{Ig}(t),W_F(t)\}$ 
are independent standard Wiener processes defined on a complete probability space 
$(\Omega,\mathcal{F},\mathbb{P})$.


The model parameters are defined as follows: $\beta>0$ denotes the rumor
transmission rate, $\sigma>0$ is the activation rate at which exposed
individuals become active spreaders, $\gamma>0$ represents the rate at which
spreaders lose interest and stop disseminating the rumor, $\rho>0$ is the rate
at which spreaders become skeptical and move into the ignorant class, and
$\theta>0$ denotes the verification or fact-checking rate. The constant
$\tau>0$ represents the behavioral and institutional response delay. Finally,
the nonnegative parameters $\sigma_S,\sigma_E,\sigma_I,\sigma_R,\sigma_{Ig}$,
and $\sigma_F$ represent the stochastic noise intensities associated with each
compartment, capturing random fluctuations in social interactions and rumor
exposure.

The stochastic perturbation terms model uncertainty arising from unpredictable online activity, irregular exposure to misinformation, and random behavioral changes in individuals during crisis situations.

\subsection{Initial Conditions and Delay History}

To ensure well-defined dynamics, the model requires initial history functions over the delay interval $[-\tau,0]$. Let
\begin{equation}
\begin{split}
&(S(\xi),E(\xi),I(\xi),R(\xi),I_g(\xi),F(\xi))
=
\\& (\phi_1(\xi),\phi_2(\xi),\phi_3(\xi),\phi_4(\xi),\phi_5(\xi),\phi_6(\xi)),
 \xi\in[-\tau,0],
\end{split}
\end{equation}
where $\phi_i(\xi)\geq 0$ are continuous functions satisfying
\begin{equation}
\phi_1(\xi)+\phi_2(\xi)+\phi_3(\xi)+\phi_4(\xi)+\phi_5(\xi)+\phi_6(\xi)=N.
\end{equation}



\subsection{Well-Posedness and Positivity of Solutions}

Let $(\Omega,\mathcal{F},\{\mathcal{F}_t\}_{t\geq 0},\mathbb{P})$ be a complete filtered probability space satisfying the usual conditions, and let 
$\{W_S(t),W_E(t),W_I(t),W_R(t),W_{Ig}(t),W_F(t)\}_{t\geq 0}$ 
be independent one-dimensional standard Brownian motions adapted to $\{\mathcal{F}_t\}_{t\geq 0}$.  
Consider the stochastic delay differential system~\eqref{eq:proposed_model} with initial history
\[
X(\xi)=\Phi(\xi), \qquad \xi\in[-\tau,0],
\]
where $X(t)=(S(t),E(t),I(t),R(t),I_g(t),F(t) )^\top$ and 
$\Phi(\xi)=(\phi_1(\xi),\phi_2(\xi),\phi_3(\xi),\phi_4(\xi),\phi_5(\xi),\phi_6(\xi))^\top$
is an $\mathcal{F}_0$-measurable continuous function satisfying $\phi_i(\xi)\geq 0$ for all $\xi\in[-\tau,0]$.

The drift and diffusion coefficients of system~\eqref{eq:proposed_model} are locally Lipschitz continuous in $(X(t),X(t-\tau))$ and satisfy a linear growth condition. Hence, by standard existence and uniqueness results for stochastic delay differential equations, there exists a unique maximal $\{\mathcal{F}_t\}$-adapted continuous solution $X(t)$ defined on $[-\tau,T_e)$, where $T_e$ denotes the explosion time (see Appendix~\ref{appendix:lipschitz_local_sol}).

Furthermore, using a standard localization argument together with It\^{o}'s formula and the linear growth bound, one can show that $T_e=\infty$ almost surely (see Appendix~\ref{appendix:global_solution}). Therefore, system~\eqref{eq:proposed_model} admits a unique global strong solution for all $t\geq 0$.

Finally, the nonnegative orthant $\mathbb{R}_+^6$ is positively invariant. In particular, since each diffusion term is of multiplicative type (i.e., proportional to the corresponding state variable), the boundary of $\mathbb{R}_+^6$ is absorbing. Hence, for any admissible nonnegative initial history $\Phi$, the solution satisfies
\[
S(t),E(t),I(t),R(t),I_g(t),F(t)\geq 0, 
\qquad \forall t\geq 0,
\quad \text{a.s}.
\]
Consequently, all state variables remain biologically and socially meaningful as population densities throughout the evolution of the system.

\subsection{Model Interpretation}

System~\eqref{eq:proposed_model} generalizes classical rumor propagation models by incorporating three key features: (i) discrete time delay in rumor influence, representing realistic hesitation and delayed institutional response; (ii) behavioral transitions into skeptical and fact-checking states, capturing voluntary resistance to misinformation; and (iii) stochastic perturbations, representing random fluctuations in exposure, user engagement, and social interaction intensity. This stochastic delay-aware framework provides a mathematically rigorous foundation for the equilibrium, stability, and optimal control analyses developed in subsequent sections.

\section{Equilibrium Analysis}

In this section, we investigate the equilibrium structure of the proposed rumor
propagation model. Since equilibrium points are defined as constant solutions of the drift system,
and since the delay does not affect constant solutions, the equilibria of
system~\eqref{eq:proposed_model} coincide with those of the associated
deterministic delay model~\eqref{eq:proposed_model}. Therefore, we analyze the
equilibria of the deterministic system.

\subsection{Equilibrium Points}

An equilibrium point 
\[
\mathcal{E}^*=(S^*,E^*,I^*,R^*,I_g^*,F^*)
\]
is defined as a constant solution of system~\eqref{eq:proposed_model}, i.e.,
a point satisfying
\[
\frac{dS}{dt}=\frac{dE}{dt}=\frac{dI}{dt}=
\frac{dR}{dt}=\frac{dI_g}{dt}=\frac{dF}{dt}=0.
\]
Since at equilibrium we have $I(t-\tau)=I(t)=I^*$, the delay term vanishes and
the equilibrium conditions are obtained by setting the right-hand sides of
\eqref{eq:proposed_model} equal to zero:
\begin{equation}
\begin{aligned}
0 &= -\beta S^* I^*, \\
0 &= \beta S^* I^* - \sigma E^*, \\
0 &= \sigma E^* - (\gamma+\rho)I^*, \\
0 &= \gamma I^*, \\
0 &= \rho I^* - \theta I_g^*, \\
0 &= \theta I_g^*.
\end{aligned}
\label{eq:equilibrium_conditions}
\end{equation}

From the fourth equation in~\eqref{eq:equilibrium_conditions}, since $\gamma>0$,
it follows immediately that
\begin{equation}
I^*=0.
\label{eq:I_star_zero}
\end{equation}
Substituting $I^*=0$ into the third equation yields $E^*=0$, and the fifth and
sixth equations imply $I_g^*=0$. Therefore, every equilibrium point satisfies
\begin{equation}
E^*=I^*=I_g^*=0.
\end{equation}

The remaining state variables satisfy only the population conservation law
\begin{equation}
S^*+R^*+F^*=N,
\qquad S^*,R^*,F^*\geq 0.
\label{eq:equilibrium_family}
\end{equation}

Hence, the system admits a continuum of rumor-free equilibria of the form
\begin{equation}
\begin{split}
\mathcal{E}^*(S^*,R^*,F^*)
&=
(S^*,0,0,R^*,0,F^*),
\\ S^*+R^*+F^*&=N.
\end{split}
\label{eq:equilibrium_set}
\end{equation}

\subsection{Rumor-Free Equilibrium}

The most relevant equilibrium is the \emph{rumor-free equilibrium} (RFE), which
corresponds to a population where no rumor has been transmitted and all
individuals remain susceptible. This equilibrium is given by
\begin{equation}
\mathcal{E}_0=(N,0,0,0,0,0).
\label{eq:RFE}
\end{equation}

This point represents the initial state of the population prior to rumor
introduction. The stability of $\mathcal{E}_0$ determines whether the rumor can
invade the population and generate a significant outbreak.

\subsection{Absence of Endemic Equilibrium}


This model differs from classical ones because it lacks a persistent endemic state. Due to absorbing states $R(t)$ and $F(t)$ and a lack of new susceptible individuals, the rumor propagation will always die out as the number of spreaders drops to zero. 

Consequently, the longevity of a rumor is defined not by a stable, long-term state, but by temporary surges in activity and its ultimate, cumulative reach. Specifically, the long-term behavior is determined by the asymptotic distribution of individuals across the $S$, $R$, and $F$ classes.

\subsection{Final-Size Equilibrium Set}

From~\eqref{eq:equilibrium_set}, it follows that the asymptotic behavior of the
system converges to a limiting equilibrium of the form
\[
(S_\infty,0,0,R_\infty,0,F_\infty),
\]
where $S_\infty+R_\infty+F_\infty=N$. The values $(S_\infty,R_\infty,F_\infty)$
depend on the initial history and model parameters. This limiting equilibrium
represents the final size of the rumor outbreak and quantifies the proportion
of the population that remains susceptible, becomes removed, or transitions into
the fact-checked class.

\subsection{Stochastic Stability of the Rumor-Free Equilibrium}

In this subsection, we study the stochastic stability of the rumor-free
equilibrium of the stochastic delay system~\eqref{eq:proposed_model}.
Let
\[
\mathcal{E}_0=(N,0,0,0,0,0)
\]
denote the rumor-free equilibrium of the deterministic drift system.

\subsubsection{Linearized Stochastic Delay System}

Let $s(t)=S(t)-N$ and consider small perturbations around $\mathcal{E}_0$. By
neglecting higher-order nonlinear terms, the stochastic delay system
\eqref{eq:proposed_model} can be approximated locally by the linear SDDE
subsystem in $(E,I)$:
\begin{equation}
\begin{aligned}
dE(t) &= \Big(\beta N I(t-\tau)-\sigma E(t)\Big)dt
        +\sigma_E E(t)\,dW_E(t), \\
dI(t) &= \Big(\sigma E(t)-(\gamma+\rho)I(t)\Big)dt
        +\sigma_I I(t)\,dW_I(t).
\end{aligned}
\label{eq:linear_SDDE}
\end{equation}
The remaining compartments $(R,I_g,F)$ are driven by $I(t)$ and therefore do not
affect the local stability of $\mathcal{E}_0$.

\subsubsection{Mean-Square Stability Definition}

The equilibrium $\mathcal{E}_0$ is said to be \emph{mean-square stable} if for
any initial history $\Phi\in C([-\tau,0],\mathbb{R}^2)$,
\[
\lim_{t\to\infty}\mathbb{E}\Big[E(t)^2+I(t)^2\Big]=0.
\]

\subsubsection{Mean-Square Stability Theorem}

\begin{theorem}
Consider the linearized stochastic delay subsystem~\eqref{eq:linear_SDDE}. If
\begin{equation}
\frac{\beta N}{\gamma+\rho}
<
1-\frac{\sigma_I^2}{2(\gamma+\rho)},
\label{eq:stochastic_threshold}
\end{equation}
then the rumor-free equilibrium $\mathcal{E}_0$ is exponentially stable in the
mean-square sense. In particular,
\[
\lim_{t\to\infty}\mathbb{E}\Big[E(t)^2+I(t)^2\Big]=0,
\qquad \forall \tau\geq 0.
\]
\label{thm:stochastic_stability}
\end{theorem}

\begin{proof}
Define the Lyapunov function
\[
V(t)=E(t)^2+cI(t)^2,
\]
where $c>0$ is a constant to be chosen. Applying It\^{o}'s formula yields
\begin{align}
d(E^2)
&=2E\,dE+(dE)^2 \nonumber \\
&=2E\Big(\beta N I_\tau-\sigma E\Big)dt
+2\sigma_EE^2\,dW_E
+\sigma_E^2E^2dt,
\label{eq:Ito_E2}
\end{align}
where $I_\tau=I(t-\tau)$. Similarly,
\begin{align}
d(I^2)
&=2I\,dI+(dI)^2 \nonumber \\
&=2I\Big(\sigma E-(\gamma+\rho)I\Big)dt
+2\sigma_II^2\,dW_I
+\sigma_I^2I^2dt.
\label{eq:Ito_I2}
\end{align}

Multiplying~\eqref{eq:Ito_I2} by $c$ and adding to~\eqref{eq:Ito_E2}, we obtain
\begin{align}
dV(t)
&=\Big[2\beta N EI_\tau-2\sigma E^2+\sigma_E^2E^2 \nonumber \\
&\quad +2c\sigma EI-2c(\gamma+\rho)I^2+c\sigma_I^2I^2\Big]dt
+dM(t),
\label{eq:dV}
\end{align}
where $M(t)$ is a local martingale term.

Using Young's inequality, for any $\varepsilon_1,\varepsilon_2>0$,
\begin{align}
2\beta N EI_\tau
&\leq \varepsilon_1E^2+\frac{\beta^2N^2}{\varepsilon_1}I_\tau^2, \\
2c\sigma EI
&\leq c\varepsilon_2E^2+\frac{c\sigma^2}{\varepsilon_2}I^2.
\end{align}

Substituting these bounds into~\eqref{eq:dV} gives
\begin{align}
&dV(t)
\leq
\Big[
-(2\sigma-\sigma_E^2-\varepsilon_1-c\varepsilon_2)E^2 \nonumber \\
&
-\Big(2c(\gamma+\rho)-c\sigma_I^2-\frac{c\sigma^2}{\varepsilon_2}\Big)I^2
+\frac{\beta^2N^2}{\varepsilon_1}I_\tau^2
\Big]dt
+dM(t).
\label{eq:dV_bound}
\end{align}

Choosing $c>0$ and $\varepsilon_2$ such that
\[
2c(\gamma+\rho)-c\sigma_I^2-\frac{c\sigma^2}{\varepsilon_2}>0,
\]
and then selecting $\varepsilon_1>0$ sufficiently small, we obtain constants
$\kappa_1,\kappa_2>0$ such that
\begin{equation}
\mathbb{E}[dV(t)]
\leq
-\kappa_1\mathbb{E}[E(t)^2]dt
-\kappa_2\mathbb{E}[I(t)^2]dt
+\kappa_3\mathbb{E}[I(t-\tau)^2]dt,
\label{eq:mean_bound}
\end{equation}
for some $\kappa_3>0$.



Since the coefficients are constant and the system is linear, the delay term can
be treated using a Razumikhin-type approach. In particular, let $V(t)$ be the
quadratic Lyapunov function introduced above and assume that there exists a
constant $q>1$ such that the Razumikhin condition holds:
\begin{equation}
V(t-\tau)\leq q\,V(t),
\qquad \forall t\geq 0.
\label{eq:raz_condition}
\end{equation}
Condition~\eqref{eq:raz_condition} restricts the influence of the delayed state
by requiring that the Lyapunov functional evaluated at the delayed time does not
dominate its current value. Such an assumption is standard in Razumikhin theory
for delay differential systems and is used to ensure that the delay feedback
remains bounded relative to the instantaneous state.

Using the definition of $V(t)$, there exists a constant $c_0>0$ such that
\begin{equation}
I(t)^2 \leq c_0 V(t),
\qquad 
I(t-\tau)^2 \leq c_0 V(t-\tau).
\label{eq:I_bound_V}
\end{equation}
Hence, by~\eqref{eq:raz_condition} we obtain
\begin{equation}
I(t-\tau)^2 \leq c_0 q V(t).
\label{eq:delay_bound}
\end{equation}

Substituting~\eqref{eq:delay_bound} into~\eqref{eq:mean_bound} yields
\begin{equation}
\mathbb{E}[dV(t)]
\leq
-\kappa_1\mathbb{E}[E(t)^2]dt
-\kappa_2\mathbb{E}[I(t)^2]dt
+\kappa_3 c_0 q\,\mathbb{E}[V(t)]dt.
\label{eq:mean_bound_raz}
\end{equation}

Furthermore, since $V(t)$ is positive definite, there exist constants
$m_1,m_2>0$ such that
\begin{equation}
m_1\big(E(t)^2+I(t)^2\big)
\leq V(t)
\leq m_2\big(E(t)^2+I(t)^2\big).
\label{eq:V_equiv}
\end{equation}
Using~\eqref{eq:V_equiv}, inequality~\eqref{eq:mean_bound_raz} can be rewritten
as
\begin{equation}
\frac{d}{dt}\mathbb{E}[V(t)]
\leq
-\Big(\delta - \kappa_3 c_0 q\Big)\mathbb{E}[V(t)],
\label{eq:V_decay}
\end{equation}
where $\delta>0$ is a constant depending on $\kappa_1,\kappa_2$ and the
equivalence bounds in~\eqref{eq:V_equiv}.

Therefore, if the parameters satisfy the condition
\begin{equation}
\kappa_3 c_0 q < \delta,
\label{eq:raz_stability_cond}
\end{equation}
then~\eqref{eq:V_decay} implies exponential decay of the mean-square Lyapunov
function:
\begin{equation}
\mathbb{E}[V(t)]
\leq
\mathbb{E}[V(0)]
\exp\!\left(-(\delta-\kappa_3 c_0 q)t\right),
\qquad t\geq 0.
\end{equation}
Consequently,
\[
\lim_{t\to\infty}\mathbb{E}[E(t)^2]
=
\lim_{t\to\infty}\mathbb{E}[I(t)^2]
=0,
\]
which establishes mean-square stability of the rumor-free equilibrium.

The Razumikhin argument shows that the destabilizing contribution of the delayed
term $I(t-\tau)$ can be controlled provided its weighted influence remains
dominated by the instantaneous dissipation terms. In particular, condition
\eqref{eq:raz_stability_cond} implies that the delay does not destroy
stability when the removal rates $(\gamma+\rho)$ and the exposure transition
rate $\sigma$ dominate the effective delayed transmission intensity.

\end{proof}

\begin{remark}
Condition~\eqref{eq:stochastic_threshold} provides a sufficient stability
criterion showing that stochastic perturbations reduce the effective invasion
capacity of the rumor. In particular, increasing noise intensity $\sigma_I$
raises the stability margin and can suppress rumor outbreaks even when the
deterministic threshold is close to unity.
\end{remark}

\section{Experimental Results and Discussion}

In this section, we present numerical experiments to validate the theoretical
analysis and to illustrate the impact of stochastic perturbations and time delay
on rumor propagation. 


Numerical approximations used stored data for the delayed term, and any negative results caused by stochasticity were projected onto the non-negative domain. To assess uncertainty, we run 100 Monte Carlo simulations to calculate mean trajectories and 95\% pointwise confidence intervals. Model validation is based on the peak spreader proportion $I(t)$, final outbreak size $R(T), F(T)$, and the confidence bands, which represent stochastic variability.



Figure \ref{fig:I_mean_ci}
depicts the temporal dynamics of the spreader class
derived from 100 Monte Carlo simulations. The average trajectory is represented by a solid curve, while the shaded region highlights the 95\% confidence interval.

\begin{figure}[!t]
\centering
\includegraphics[width=\columnwidth]{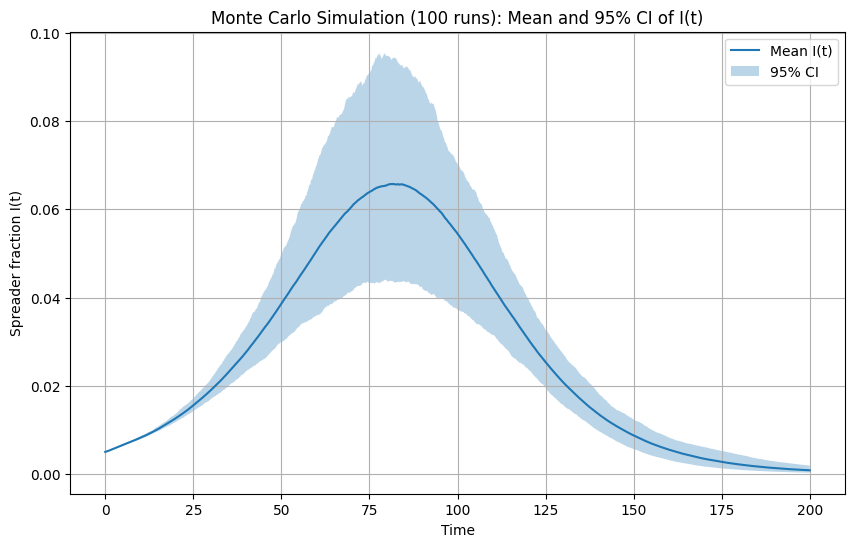}
\caption{Mean spreader population $I(t)$
(100-run Monte Carlo simulation) with 95\% confidence intervals, showing a sharp outbreak peak and subsequent decay. Stochastic perturbations drive increased variability around the peak.}
\label{fig:I_mean_ci}
\end{figure}


Let $E(t)$ and $S(t)$ denote the proportions of individuals in the
\emph{exposed} and \emph{active spreader} states at time $t$. Rumor growth
is driven by the transition
\[
E \xrightarrow{\beta} S,
\]
which can be modeled as
\begin{equation}
\frac{dS(t)}{dt} = \beta E(t) - \gamma S(t),
\end{equation}
where $\beta>0$ is the activation rate and $\gamma>0$ is the stifling rate.

The spreader population typically exhibits a unimodal trajectory with a
dominant peak
\[
t^\ast = \arg\max_{t\ge0} S(t),
\]
after which
\[
\lim_{t\to\infty} S(t)=0 .
\]

This asymptotic behavior indicates convergence to the rumor-free
equilibrium
\[
\mathcal{E}_0 = \{(E,S): S=0\},
\]
which is stable when the effective reproduction number $\mathcal{R}_0<1$,
implying the absence of a persistent rumor-endemic state.

When stochastic fluctuations are considered, the spreader dynamics can be
written as
\begin{equation}
dS(t) = \left[\beta E(t)-\gamma S(t)\right]dt + \sigma S(t)dW(t),
\end{equation}
where $\sigma>0$ denotes the noise intensity and $W(t)$ is a Wiener
process. The variance of $S(t)$ increases near $t^\ast$, producing a wider
confidence band around the peak, which indicates that random contact
fluctuations have the strongest influence during the rapid rumor growth
phase.

Fig.~\ref{fig:mean_all} shows the mean trajectories of all compartments
$S(t),E(t),I(t),R(t),I_g(t)$, and $F(t)$ over the simulation horizon.

\begin{figure}[!t]
\centering
\includegraphics[width=\columnwidth]{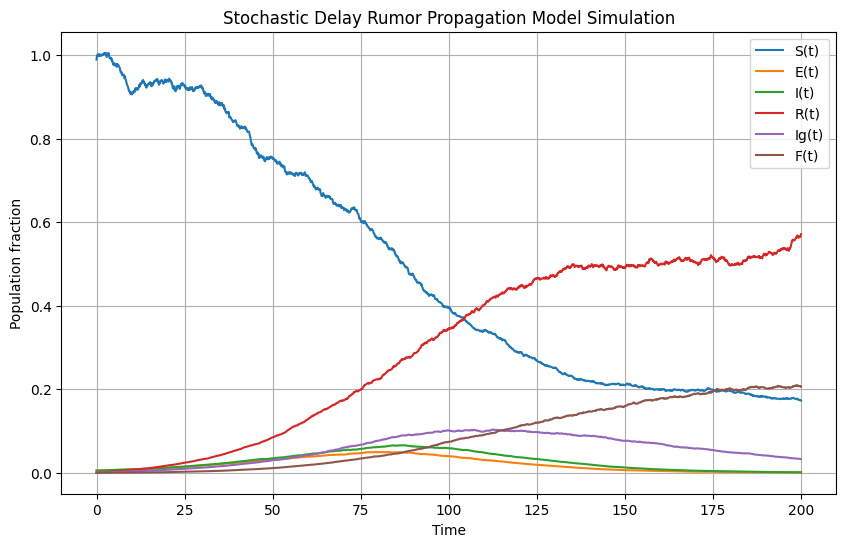}
\caption{The average trajectory of all compartments shows the spreader population $I(t)$
rising to an outbreak peak before declining. Meanwhile, both the removed and fact-checked populations demonstrate continuous accumulation over time.}
\label{fig:mean_all}
\end{figure}

The plots demonstrate that the susceptible population $S(t)$ decreases
monotonically as individuals become exposed through delayed contact with
spreader history. The exposed population $E(t)$ follows a similar outbreak
pattern as $I(t)$ but with an earlier peak, reflecting the latency stage of the
rumor dynamics. The removed class $R(t)$ increases steadily and eventually
captures a large proportion of the population, representing individuals who have
lost interest or stopped spreading the rumor. The skeptical class $I_g(t)$ rises
during the outbreak phase and later decreases due to the transition into the
fact-checking compartment. Finally, the fact-checked class $F(t)$ grows
monotonically, showing the cumulative impact of verification processes in
reducing rumor influence.

Table~\ref{tab:ablation_delay_R0} summarizes the ablation study on the combined effects of information delay $\tau$ and basic reproduction number $R_0$ on rumor dynamics. Two key performance indicators are presented: the mean maximum value of the spreader population, $\max I(t)$, and the final mean size of the epidemic, measured by $R(T)+F(T)$, along with their respective standard deviations calculated from Monte Carlo simulations.


First, for a fixed time delay $\tau$, increasing $R_0$ leads to a marked and monotonic rise in both the peak and final outbreak size. When $R_0 < 1$, the proportion of people spreading the rumor at the peak remains close to its initial value, and the final outbreak size is small, indicating that the rumor is not spreading significantly. When $R_0$ exceeds $1$, an abrupt transition is observed: the peak of $I(t)$ increases rapidly, and the final epidemic size reaches values exceeding 70\%–80\% of the population for $R_0 = 2.0$. This behavior is consistent with the theoretical threshold role of $R_0$, confirming that higher transmission intensity leads to more severe and widespread rumor outbreaks.


Secondly, for a given $R_0$, increasing the delay $\tau$ generally amplifies the spread of rumors. While peak spread values are only slightly affected for $R_0 < 1$, the effect of the delay becomes significant for $R_0 > 1$. In this case, longer delays lead to greater variability and, often, slightly larger epidemic peaks and final sizes. This indicates that the delay in information processing and response allows the rumor to circulate longer before corrective mechanisms kick in, thus increasing the magnitude and uncertainty of the epidemic.


Finally, the reported standard deviations highlight the increasing influence of stochastic effects as $R_0$ and $\tau$ increase. In particular, for large values of $R_0$ and significant delays, the variability in the final magnitude of the outbreak size becomes substantial, demonstrating that random fluctuations in social interactions can significantly alter the course of the infodemic, under identical parameter settings. Overall, these results corroborate the analytical findings and underscore the crucial importance of reducing $R_0$ and minimizing information dissemination delays to limit the spread of rumors.

\begin{table}[!t]
\centering
\caption{Ablation study of delay $\tau$ and basic reproduction number $R_0$.
Reported values are Monte Carlo means with standard deviations (Std).}
\label{tab:ablation_delay_R0}
\resizebox{\columnwidth}{!}{%
\begin{tabular}{|c|c|c|c|c|}
\hline
$\tau$ & $R_0$ & $\beta$ & Peak $I(t)$ (Mean $\pm$ Std) & Final Size $R(T)+F(T)$ (Mean $\pm$ Std) \\
\hline
0  & 0.5 & 0.075 & $0.00527 \pm 0.00006$ & $0.0194 \pm 0.0020$ \\
0  & 0.8 & 0.120 & $0.00544 \pm 0.00007$ & $0.0465 \pm 0.0085$ \\
0  & 1.0 & 0.150 & $0.00693 \pm 0.00350$ & $0.1233 \pm 0.0536$ \\
0  & 1.2 & 0.180 & $0.01618 \pm 0.00784$ & $0.2988 \pm 0.1191$ \\
0  & 1.5 & 0.225 & $0.04420 \pm 0.01301$ & $0.5726 \pm 0.1119$ \\
0  & 2.0 & 0.300 & $0.09704 \pm 0.01714$ & $0.8105 \pm 0.1062$ \\
\hline
5  & 0.5 & 0.075 & $0.00529 \pm 0.00008$ & $0.0234 \pm 0.0023$ \\
5  & 0.8 & 0.120 & $0.00541 \pm 0.00009$ & $0.0526 \pm 0.0095$ \\
5  & 1.0 & 0.150 & $0.00620 \pm 0.00175$ & $0.1233 \pm 0.0290$ \\
5  & 1.2 & 0.180 & $0.01331 \pm 0.00540$ & $0.2641 \pm 0.0844$ \\
5  & 1.5 & 0.225 & $0.03174 \pm 0.01058$ & $0.5230 \pm 0.1247$ \\
5  & 2.0 & 0.300 & $0.06572 \pm 0.01270$ & $0.7445 \pm 0.1110$ \\
\hline
10 & 0.5 & 0.075 & $0.00529 \pm 0.00006$ & $0.0262 \pm 0.0032$ \\
10 & 0.8 & 0.120 & $0.00540 \pm 0.00007$ & $0.0595 \pm 0.0129$ \\
10 & 1.0 & 0.150 & $0.00583 \pm 0.00066$ & $0.1224 \pm 0.0273$ \\
10 & 1.2 & 0.180 & $0.01130 \pm 0.00440$ & $0.2337 \pm 0.0592$ \\
10 & 1.5 & 0.225 & $0.02575 \pm 0.00873$ & $0.4709 \pm 0.1295$ \\
10 & 2.0 & 0.300 & $0.05448 \pm 0.01315$ & $0.7429 \pm 0.1261$ \\
\hline
\end{tabular}}
\end{table}

\section{Conclusion and Future Work}


This paper presents a stochastic differential equation model with a delay for rumor propagation,  incorporating behavioral delay, skepticism, and fact-checking.  Equilibrium analysis shows that the system admits a continuum of rumor-free equilibria,  and no positive endemic equilibrium. Local stability results,  establish that the rumor-free equilibrium is asymptotically stable for $R_0 < 1$ and unstable for $R_0 > 1$. Monte Carlo simulations validate the theoretical results and confirm that a higher $R_0$ increases both the peak of the outbreak,  and the final size of the outbreak, while an increase in the delay $\tau$ amplifies the severity of the infodemic outbreak as well as the associated uncertainty.

Future work may include extending the model to heterogeneous network
structures, estimating parameters using real social media datasets, and
investigating fractional-order, distributed-delay, or state-dependent delay
extensions to better capture long-memory effects in misinformation dynamics.

\appendix

\subsection{Local Existence and Uniqueness}
\label{appendix:lipschitz_local_sol}

Let 
\[
X(t)=\big(S(t),E(t),I(t),R(t),I_g(t),F(t)\big)^\top\in\mathbb{R}^6
\]
and define the delayed state
\[
X_\tau(t)=X(t-\tau).
\]
System~\eqref{eq:proposed_model} can be written in vector form as
\begin{equation}
dX(t)=b\big(X(t),X_\tau(t)\big)\,dt
+\Sigma\big(X(t)\big)\,dW(t),
\label{eq:vector_SDDE}
\end{equation}
where $W(t)\in\mathbb{R}^6$ is a standard Brownian motion and the drift vector
$b:\mathbb{R}^6\times\mathbb{R}^6\to\mathbb{R}^6$ is given by
\begin{equation}
b(x,y)=
\begin{pmatrix}
-\beta x_1y_3 \\
\beta x_1y_3-\sigma x_2 \\
\sigma x_2-(\gamma+\rho)x_3 \\
\gamma x_3 \\
\rho x_3-\theta x_5 \\
\theta x_5
\end{pmatrix},
\label{eq:drift_def}
\end{equation}
with $x=(x_1,\dots,x_6)^\top$ and $y=(y_1,\dots,y_6)^\top$. The diffusion matrix
$\Sigma:\mathbb{R}^6\to\mathbb{R}^{6\times 6}$ is diagonal:
\begin{equation}
\Sigma(x)=
\mathrm{diag}\big(
\sigma_S x_1,\,
\sigma_E x_2,\,
\sigma_I x_3,\,
\sigma_R x_4,\,
\sigma_{Ig} x_5,\,
\sigma_F x_6
\big).
\label{eq:diffusion_def}
\end{equation}

\paragraph{Local Lipschitz continuity.}
Fix $R>0$ and consider $x,\bar{x},y,\bar{y}\in\mathbb{R}^6$ satisfying
\[
\|x\|\leq R,\quad \|\bar{x}\|\leq R,\quad
\|y\|\leq R,\quad \|\bar{y}\|\leq R.
\]
Then, using the inequality $|ab-\bar{a}\bar{b}|
\leq |a-\bar{a}||b|+|\bar{a}||b-\bar{b}|$, we obtain
\begin{align*}
|x_1y_3-\bar{x}_1\bar{y}_3|
&\leq |x_1-\bar{x}_1||y_3|
+|\bar{x}_1||y_3-\bar{y}_3| \\
&\leq R|x_1-\bar{x}_1|+R|y_3-\bar{y}_3|.
\end{align*}
Hence, for the first drift component,
\[
|b_1(x,y)-b_1(\bar{x},\bar{y})|
=\beta|x_1y_3-\bar{x}_1\bar{y}_3|
\leq \beta R\big(|x_1-\bar{x}_1|+|y_3-\bar{y}_3|\big).
\]
Similarly, for the second component,
\begin{align*}
|b_2(x,y)-b_2(\bar{x},\bar{y})|
&=
\big|\beta x_1y_3-\sigma x_2
-\beta \bar{x}_1\bar{y}_3+\sigma \bar{x}_2\big| \\
&\leq
\beta|x_1y_3-\bar{x}_1\bar{y}_3|
+\sigma|x_2-\bar{x}_2| \\
&\leq
\beta R(|x_1-\bar{x}_1|+|y_3-\bar{y}_3|)
+\sigma|x_2-\bar{x}_2|.
\end{align*}
The remaining components are linear and satisfy analogous estimates. Therefore,
there exists a constant $L_R>0$ such that
\begin{equation}
\begin{split}
&\|b(x,y)-b(\bar{x},\bar{y})\|
\leq
L_R\big(\|x-\bar{x}\|+\|y-\bar{y}\|\big),
\\
&\forall \|x\|,\|\bar{x}\|,\|y\|,\|\bar{y}\|\leq R.  
\end{split}
\label{eq:drift_Lip}
\end{equation}

For the diffusion coefficient, since $\Sigma$ is diagonal,
\[
\|\Sigma(x)-\Sigma(\bar{x})\|^2
=\sum_{i=1}^6 \sigma_i^2(x_i-\bar{x}_i)^2
\leq C\|x-\bar{x}\|^2,
\]
where $C=\max\{\sigma_S^2,\sigma_E^2,\sigma_I^2,
\sigma_R^2,\sigma_{Ig}^2,\sigma_F^2\}$. Thus,
\begin{equation}
\|\Sigma(x)-\Sigma(\bar{x})\|
\leq
\sqrt{C}\,\|x-\bar{x}\|,
\label{eq:diff_Lip}
\end{equation}
which shows that $\Sigma$ is globally Lipschitz and hence locally Lipschitz.

\paragraph{Linear growth condition.}
Using Young's inequality $2ab\leq a^2+b^2$, we estimate
\[
|x_1y_3|
\leq \tfrac12(x_1^2+y_3^2)
\leq \tfrac12(\|x\|^2+\|y\|^2).
\]
Therefore,
\[
\|b(x,y)\|^2
\leq C_1\big(1+\|x\|^2+\|y\|^2\big),
\]
for some constant $C_1>0$ depending only on the parameters. Similarly,
\[
\|\Sigma(x)\|^2
=\sum_{i=1}^6 \sigma_i^2 x_i^2
\leq C_2\|x\|^2
\leq C_2\big(1+\|x\|^2+\|y\|^2\big),
\]
where $C_2=\max\{\sigma_S^2,\sigma_E^2,\sigma_I^2,
\sigma_R^2,\sigma_{Ig}^2,\sigma_F^2\}$.

Hence, there exists $K>0$ such that
\begin{equation}
\|b(x,y)\|^2+\|\Sigma(x)\|^2
\leq K\big(1+\|x\|^2+\|y\|^2\big),
\qquad \forall x,y\in\mathbb{R}^6.
\label{eq:linear_growth}
\end{equation}

\paragraph{Existence and uniqueness.}
From \eqref{eq:drift_Lip}, \eqref{eq:diff_Lip}, and \eqref{eq:linear_growth},
the coefficients of system~\eqref{eq:vector_SDDE} satisfy local Lipschitz
continuity and the linear growth condition. Therefore, by standard existence and
uniqueness results for stochastic delay differential equations, there exists a
unique maximal $\{\mathcal{F}_t\}$-adapted continuous strong solution $X(t)$
defined on $[-\tau,T_e)$, where $T_e$ denotes the explosion time.
\hfill $\blacksquare$

\subsection{Global Existence (Non-Explosion)}
\label{appendix:global_solution}

Let 
\[
X(t)=\big(S(t),E(t),I(t),R(t),I_g(t),F(t)\big)^\top
\]
denote the solution of system~\eqref{eq:proposed_model}. Since the drift
and diffusion coefficients are locally Lipschitz continuous, there exists a
unique maximal local strong solution $X(t)$ defined on $[-\tau,T_e)$, where
$T_e$ denotes the explosion time.

We now show that $T_e=\infty$ almost surely.

\paragraph{Step 1: Stopping times.}
For each $n\in\mathbb{N}$, define the stopping time
\begin{equation}
\tau_n=\inf\Big\{t\geq 0:\|X(t)\|\geq n\Big\},
\end{equation}
with the convention $\inf\emptyset=\infty$. Clearly, $\{\tau_n\}_{n\geq 1}$ is
an increasing sequence and
\[
\tau_n\uparrow T_e \quad \text{as } n\to\infty,
\qquad \text{a.s.}
\]

\paragraph{Step 2: Lyapunov function.}
Consider the quadratic Lyapunov function
\begin{equation}
\begin{split}
V(X(t))&=\|X(t)\|^2\\
&=S(t)^2+E(t)^2+I(t)^2+R(t)^2+I_g(t)^2+F(t)^2.
\end{split}
\end{equation}

Applying It\^{o}'s formula to $V(X(t\wedge\tau_n))$ yields
\begin{equation}
dV(X(t))
=\mathcal{L}V(X(t),X(t-\tau))dt+dM(t),
\label{eq:ItoV_column}
\end{equation}
where $M(t)$ is a local martingale and $\mathcal{L}$ is the infinitesimal
generator of the SDDE.

\paragraph{Step 3: Generator estimate.}
Let $I_\tau=I(t-\tau)$. Using \eqref{eq:proposed_model}, we obtain
\begin{align}
\mathcal{L}V
&=2S(-\beta SI_\tau)
+2E(\beta SI_\tau-\sigma E) \nonumber \\
&\quad +2I(\sigma E-(\gamma+\rho)I)
+2R(\gamma I) \nonumber \\
&\quad +2I_g(\rho I-\theta I_g)
+2F(\theta I_g) \nonumber \\
&\quad +\sigma_S^2S^2+\sigma_E^2E^2+\sigma_I^2I^2
+\sigma_R^2R^2 \nonumber \\
&\quad +\sigma_{Ig}^2I_g^2+\sigma_F^2F^2.
\label{eq:LV_expand_column}
\end{align}

Using Young's inequality $2ab\leq a^2+b^2$, we estimate the cross terms as
\begin{align*}
2\beta SEI_\tau &\leq \beta(S^2+E^2I_\tau^2)
\leq \beta(S^2+E^2+I_\tau^2), \\
2\sigma IE &\leq \sigma(I^2+E^2), \qquad
2\gamma RI \leq \gamma(R^2+I^2), \\
2\rho I_gI &\leq \rho(I_g^2+I^2), \qquad
2\theta FI_g \leq \theta(F^2+I_g^2).
\end{align*}

Hence, there exists a constant $C>0$ such that
\begin{equation}
\mathcal{L}V(X(t),X(t-\tau))
\leq
C\Big(1+\|X(t)\|^2+\|X(t-\tau)\|^2\Big).
\label{eq:LV_bound_column}
\end{equation}

\paragraph{Step 4: Razumikhin-type estimate for the delay term.}
To handle the delayed contribution $\|X(t-\tau)\|^2$, we employ a Razumikhin-type
argument. Assume that there exists a constant $q>1$ such that the Razumikhin
condition holds:
\begin{equation}
V(X(t-\tau))\leq q\,V(X(t)),
\qquad \forall t\geq 0.
\label{eq:razumikhin_cond_global}
\end{equation}
Condition~\eqref{eq:razumikhin_cond_global} is standard in the stability theory
of stochastic delay systems and ensures that the delayed state remains bounded
relative to the current state.

Using \eqref{eq:razumikhin_cond_global} in \eqref{eq:LV_bound_column}, we obtain
\begin{equation}
\mathcal{L}V(X(t),X(t-\tau))
\leq
C\Big(1+(1+q)V(X(t))\Big).
\label{eq:LV_bound_raz}
\end{equation}

\paragraph{Step 5: Moment bound.}
Integrating \eqref{eq:ItoV_column} from $0$ to $t\wedge\tau_n$ yields
\begin{equation}
\begin{split}
V(X(t\wedge\tau_n))
&=V(X(0))
+\int_0^{t\wedge\tau_n}\!\!\!\!\mathcal{L}V(X(s),X(s-\tau))\,ds
\\&+M(t\wedge\tau_n).
\end{split}
\end{equation}

Taking expectation and using the martingale property
$\mathbb{E}[M(t\wedge\tau_n)]=0$, together with \eqref{eq:LV_bound_raz}, we
obtain
\begin{align}
\mathbb{E}\big[V(X(t\wedge\tau_n))\big]
&\leq
\mathbb{E}[V(X(0))]
+C t \nonumber \\
&\quad +C(1+q)\int_0^t
\mathbb{E}\big[V(X(s\wedge\tau_n))\big]ds.
\label{eq:Gronwall_column}
\end{align}

By Gr\"onwall's inequality, it follows that
\begin{equation}
\mathbb{E}\big[V(X(t\wedge\tau_n))\big]
\leq
\Big(\mathbb{E}[V(X(0))]+Ct\Big)
\exp\!\big(C(1+q)t\big),
\qquad t\geq 0,
\label{eq:moment_bound_column}
\end{equation}
which provides a uniform moment bound independent of $n$.

\paragraph{Step 6: Non-explosion.}
Assume by contradiction that $\mathbb{P}(T_e<\infty)>0$. Then there exist
$T>0$ and $\varepsilon>0$ such that
\[
\mathbb{P}(T_e\leq T)\geq \varepsilon.
\]
Since $\tau_n\uparrow T_e$, there exists $n_0$ such that for all $n\geq n_0$,
\[
\mathbb{P}(\tau_n\leq T)\geq \frac{\varepsilon}{2}.
\]

On the event $\{\tau_n\leq T\}$ we have $\|X(\tau_n)\|\geq n$, hence
$V(X(\tau_n))\geq n^2$. Therefore,
\begin{align}
\mathbb{E}\big[V(X(T\wedge\tau_n))\big]
&\geq
\mathbb{E}\Big[V(X(\tau_n))\mathbf{1}_{\{\tau_n\leq T\}}\Big] \nonumber \\
&\geq
n^2\mathbb{P}(\tau_n\leq T)
\geq
\frac{\varepsilon}{2}n^2.
\end{align}

Letting $n\to\infty$ implies that
$\mathbb{E}[V(X(T\wedge\tau_n))]\to\infty$, which contradicts the uniform bound
\eqref{eq:moment_bound_column}. Hence,
\[
\mathbb{P}(T_e=\infty)=1.
\]

Therefore, system~\eqref{eq:proposed_model} admits a unique global
strong solution $X(t)$ defined for all $t\geq 0$ almost surely.
\hfill $\blacksquare$

\bibliographystyle{IEEEtran}
\bibliography{main}

@article{ishizumi2024beyond,
  title={Beyond misinformation: developing a public health prevention framework for managing information ecosystems},
  author={Ishizumi, Atsuyoshi and Kolis, Jessica and Abad, Neetu and Prybylski, Dimitri and Brookmeyer, Kathryn A and Voegeli, Christopher and Wardle, Claire and Chiou, Howard},
  journal={The Lancet Public Health},
  volume={9},
  number={6},
  pages={e397--e406},
  year={2024},
  publisher={Elsevier}
}

@article{caceres2022impact,
  title={The impact of misinformation on the COVID-19 pandemic},
  author={Caceres, Maria Mercedes Ferreira and Sosa, Juan Pablo and Lawrence, Jannel A and Sestacovschi, Cristina and Tidd-Johnson, Atiyah and Rasool, Muhammad Haseeb UI and Gadamidi, Vinay Kumar and Ozair, Saleha and Pandav, Krunal and Cuevas-Lou, Claudia and others},
  journal={AIMS public health},
  volume={9},
  number={2},
  pages={262},
  year={2022}
}

@incollection{world2024infodemic,
  title={Infodemic management: protecting people from harmful health information in emergencies},
  author={World Health Organization and others},
  booktitle={Infodemic management: protecting people from harmful health information in emergencies},
  year={2024}
}

@article{tangcharoensathien2020framework,
  title={Framework for managing the COVID-19 infodemic: methods and results of an online, crowdsourced WHO technical consultation},
  author={Tangcharoensathien, Viroj and Calleja, Neville and Nguyen, Tim and Purnat, Tina and D’Agostino, Marcelo and Garcia-Saiso, Sebastian and Landry, Mark and Rashidian, Arash and Hamilton, Clayton and AbdAllah, Abdelhalim and others},
  journal={Journal of medical Internet research},
  volume={22},
  number={6},
  pages={e19659},
  year={2020},
  publisher={JMIR Publications Inc., Toronto, Canada}
}

@book{tomes2022historical,
  title={What are the historical roots of the COVID-19 infodemic? Lessons from the past},
  author={Tomes, Nancy and Parry, Manon S},
  year={2022},
  publisher={World Health Organization. Regional Office for Europe}
}

@misc{gisondi2022deadly,
  title={A deadly infodemic: social media and the power of COVID-19 misinformation},
  author={Gisondi, Michael A and Barber, Rachel and Faust, Jemery Samuel and Raja, Ali and Strehlow, Matthew C and Westafer, Lauren M and Gottlieb, Michael},
  journal={Journal of medical Internet research},
  volume={24},
  number={2},
  pages={e35552},
  year={2022},
  publisher={JMIR Publications Toronto, Canada}
}

@article{do2022infodemics,
  title={Infodemics and health misinformation: a systematic review of reviews},
  author={Do Nascimento, Israel Junior Borges and Pizarro, Ana Beatriz and Almeida, Jussara M and Azzopardi-Muscat, Natasha and Gon{\c{c}}alves, Marcos Andr{\'e} and Bj{\"o}rklund, Maria and Novillo-Ortiz, David},
  journal={Bulletin of the World Health Organization},
  volume={100},
  number={9},
  pages={544},
  year={2022}
}

@article{govindankutty2024epidemic,
  title={Epidemic modeling for misinformation spread in digital networks through a social intelligence approach},
  author={Govindankutty, Sreeraag and Gopalan, Shynu Padinjappurath},
  journal={Scientific Reports},
  volume={14},
  number={1},
  pages={19100},
  year={2024},
  publisher={Nature Publishing Group UK London}
}

@article{daley1965stochastic,
  author = {D. J. Daley and D. G. Kendall},
  title = {Stochastic rumours},
  journal = {Journal of the Institute of Mathematics and Its Applications},
  volume = {1},
  number = {1},
  pages = {42--55},
  year = {1965}
}

@article{li2021social,
  title={Social media rumor refutation effectiveness: Evaluation, modelling and enhancement},
  author={Li, Zongmin and Zhang, Qi and Du, Xinyu and Ma, Yanfang and Wang, Shihang},
  journal={Information Processing \& Management},
  volume={58},
  number={1},
  pages={102420},
  year={2021},
  publisher={Elsevier}
}

@article{zhao2021detecting,
  title={Detecting health misinformation in online health communities: Incorporating behavioral features into machine learning based approaches},
  author={Zhao, Yuehua and Da, Jingwei and Yan, Jiaqi},
  journal={Information Processing \& Management},
  volume={58},
  number={1},
  pages={102390},
  year={2021},
  publisher={Elsevier}
}

@article{czerniak2023scoping,
  title={A scoping review of digital health interventions for combating COVID-19 misinformation and disinformation},
  author={Czerniak, Katarzyna and Pillai, Raji and Parmar, Abhi and Ramnath, Kavita and Krocker, Joseph and Myneni, Sahiti},
  journal={Journal of the American Medical Informatics Association},
  volume={30},
  number={4},
  pages={752--760},
  year={2023},
  publisher={Oxford University Press}
}

@article{daley1964epidemics,
  title={Epidemics and rumours},
  author={Daley, Daryl J and Kendall, David G},
  journal={Nature},
  volume={204},
  number={4963},
  pages={1118--1118},
  year={1964},
  publisher={Nature Publishing Group UK London}
}

@article{sun2021uncertain,
  title={An uncertain SIR rumor spreading model},
  author={Sun, Hang and Sheng, Yuhong and Cui, Qing},
  journal={Advances in Difference Equations},
  volume={2021},
  number={1},
  pages={286},
  year={2021},
  publisher={Springer}
}

\end{document}